\documentclass[letterpaper,10 pt,conference]{ieeeconf}  
\IEEEoverridecommandlockouts                              

\let\labelindent\relax

\usepackage{comment}

\usepackage{subfigure}

\usepackage{amsmath, amssymb,bbm,amsthm,bbold}
\usepackage{xcolor}
\usepackage{xparse}
\usepackage{mathrsfs}
\usepackage{graphicx}
\usepackage{enumitem} 
\usepackage{hyperref}
\hypersetup{colorlinks=true,  linkcolor=blue,  breaklinks=true,  urlcolor=blue,  citecolor=blue}

\usepackage{soul}
\usepackage{url}

\newtheorem{theorem}{Theorem}
\newtheorem{remark}[theorem]{Remark}

\newtheorem{lemma}[theorem]{Lemma}

\newtheorem{proposition}[theorem]{Proposition}

\newtheorem{assumption}{Assumption}

\allowdisplaybreaks

\title{\LARGE\bf  Distributed State Estimation of Discrete-Time LTI Systems via Jordan Canonical Representation}

\author{Giulio Fattore, Maria Elena Valcher, Rui Gao and Guang-Hong Yang 
\thanks{G. Fattore and  M.E. Valcher
are with the Dipartimento di Ingegneria dell'Informazione, Universit\`a di Padova, via Gradenigo 6B, 35131 Padova, Italy, e-mail:  \texttt{fattoregiu@dei.unipd.it, meme@dei.unipd.it}. Rui Gao and Guang-Hong Yang are with the College of Information Science and Engineering, Northeastern University Shenyang 110819 China, and also with The State Key Laboratory of Synthetical Automation for Process Industries, Northeastern University Shenyang 110819 China, e-mail:  \texttt{gaorui@ise.neu.edu.cn,yang-guanghong@ise.neu.edu.cn.} 
}}

\date{}

\begin{document}
\maketitle

\begin{abstract}
In this paper, we address the problem of distributed state estimation for a discrete-time, linear time-invariant system. Building on the framework proposed in \cite{GaoYang}, we exploit the Jordan canonical form of the system matrix to develop a distributed estimation scheme that ensures the asymptotic convergence of the local state estimates to the true system state. The proposed approach relies on the idea that each node reconstructs the components of the system state that are detectable for it through a local Luenberger observer, while employing a consensus-based strategy to estimate the undetectable components.  Necessary and sufficient conditions for the existence of a distributed observer that guarantees asymptotic estimation accuracy are derived. Compared with the previous work \cite{GaoYang}, the proposed design offers greater flexibility in the selection of the coupling gains and leads to a less restrictive set of conditions for solvability.
\end{abstract}

\section{Introduction}
In recent decades, the increasing complexity and scale of modern systems have brought forth numerous challenges. Among these, the distributed state estimation problem has attracted significant research interest. In this framework, a central dynamical system is observed by a network of sensors, each capable of obtaining only partial measurements of the system state and having knowledge of part (or possibly all) of the system control inputs. Since individual sensors are typically unable to estimate the full state of the system, they rely on consensus strategies, exchanging information about their local estimates with neighboring nodes.

Over the past years, several distributed state estimation algorithms have been proposed, both in the stochastic setting, using variants of the Kalman filter \cite{Tomlin,OS_distrKalman2,OS_distrKalman}, and in the deterministic setting, using Luenberger-type observers. In the continuous-time domain, the distributed estimation problem has been extensively studied under various assumptions on network connectivity and joint detectability \cite{Han2019,yang2022state}. However, the problem becomes considerably more challenging in the discrete-time setting. In fact, while continuous-time solutions typically rely on consensus strategies with high-gain coupling, these approaches do not translate directly to discrete-time systems \cite{Wang24}. To address this issue, Park and Martins \cite{Nuno1, Nuno2, Nuno3} investigated distributed state estimation for discrete-time autonomous systems using a network of Luenberger-like observers. Each local observer, equipped with an additional internal state, measures a subset of the system outputs and computes a local estimate based on its own measurements and the estimates received from its neighbors through a communication graph. By introducing an augmented state, the observer design problem is reformulated as the synthesis of a decentralized dynamic output feedback controller, enabling the use of well-established results from classical control theory. Building upon this approach, \cite{Disaro_TAC2025} extends the results to systems subject to unknown process disturbances and partial control input measurements. 

In \cite{Wang24}, a unified framework for distributed state estimation of both continuous- and discrete-time linear systems—with possibly time-varying communication graphs—has been proposed. The approach is based on partitioning the system spectrum into disjoint subsets associated with the observable and unobservable subspaces. The estimation strategy then employs a high-gain mechanism in the continuous-time case and utilizes two related sampling rates for discrete-time systems.

In \cite{MitraSundaram}, Mitra and Sundaram proposed two distinct strategies for the design of distributed observers.
In the first approach, the authors assume the  detectability of each source component of the communication graph and introduce the notion of {\em multisensor observable canonical decomposition}. This decomposition reduces the system matrix and the aggregated output matrix into a lower block-triangular structure, where each pair of diagonal blocks defines an observable subsystem. This formulation allows each sensing agent to independently estimate a specific portion of the overall system state.
The second design strategy in \cite{MitraSundaram} is based on the assumption that, for each unstable eigenvalue of the system matrix, there exists at least one root node within every source component for which that eigenvalue is observable. Under this assumption, the authors propose a coordinated Jordan canonical transformation of the system matrix (to the best of our knowledge, this represents the first use of a Jordan decomposition in the context of distributed observer design) followed by a local transformation that brings the system into the standard (Kalman) detectability form. In this case, the estimation problem is addressed through an augmented-order observer.
In both strategies, the authors assume the ability to modify the edge weights and, if necessary, remove certain connections to form a spanning tree topology. Subsequently, in \cite{MitraSundaram2019}, the same authors extended their framework to scenarios involving adversarial agents, where some nodes in the network may be compromised. Still relying on the use of a Jordan canonical transformation, Gao and Yang in \cite{GaoYang} proposed a distributed observer design framework in which, for each node, a permutation matrix brings the system into the standard (Kalman) detectability form while preserving the Jordan structure of the system matrix. 
Each $i$th agent implements a Luenberger-type observer to independently estimate the entries of the state vector that are  detectable for it, meanwhile relying on a consensus-type strategy to estimate the remaining entries, that makes use of a common coupling strength.
This enables the derivation of necessary and sufficient conditions - on both the eigenvalues of the system matrix and the communication graph -  for the solvability of the distributed estimation problem, without modifying the graph topology or edge weights.\\
For an overview of the various methodologies for the design of distributed observers for LTI systems, both deterministic and stochastic, the reader is referred to \cite{ReviewDIstrObs,Rego19}.
\medskip

Building upon the results in \cite{GaoYang}, this paper advances the state of the art in distributed state estimation by introducing several key improvements. Unlike \cite{GaoYang}, in the proposed method each agent retains, within the locally estimated portion of the state,  only the information corresponding to Jordan miniblocks that are entirely detectable.  Meanwhile, to estimate the remaining entries, the agent relies on a consensus strategy that makes use of a specific coupling strength for each remaining miniblock. This strategy simplifies the selection of the coupling gains.
In addition, allowing for a different coupling strength for each miniblock, rather than a single global value, relaxes the existence conditions and leads to a less restrictive solvability criterion. 
More details on the comparison with \cite{GaoYang} are provided in Remark \ref{comparison}.
\smallskip

The paper is organized as follows. Section \ref{sec.pb.setup} formalizes the problem setup, identifying for each agent the components of the state that are completely observable, partially observable, and unobservable, and accordingly permuting the state vectors and the describing matrices. Section \ref{sec.pb.solution} presents the proposed solution, including necessary and sufficient conditions for problem solvability. Section \ref{sec.example} provides an illustrative example that demonstrates the effectiveness of the proposed approach. Finally, Section \ref{sec.conclusion} draws the main conclusions.
\medskip

{\bf Notation.} 
The sets of real numbers, nonnegative integers, and complex numbers are denoted by $\mathbb{R}$, $\mathbb{Z}_+$, and $\mathbb{C}$, respectively.
Given two integers $h$ and $k$, with $h \le k$, we let $[h,k]$ denote the set $\{h,h+1,\dots,k\}$.
The symbols $I$ and $\mathbb{0}$ denote the identity matrix and   the zero matrix/vector of suitable dimensions, respectively.  The $i$th {\em canonical vector} whose entries are all zeros except for the $i$th, which is equal to $1$, is denoted by $\mathbb{e}_{i}$ and its dimension can be deduced from the context. 
A square matrix $P\in {\mathbb R}^{n \times n}$ obtained from the identity $I_n$ via row-column permutations is called {\em permutation matrix}.
Given any matrix $Q$, its $(i,j)$th entry is denoted by $[Q]_{i,j}$.  
The {\em Kronecker product} is denoted by $\otimes$.
Given matrices $M_i, i\in [1,p]$, the  column stacking and the row juxtaposition of these matrices are denoted by ${\rm col}\{M_i\}_{i\in[1,p]}$,  and ${\rm row}\{M_i\}_{i\in[1,p]}$, respectively, while 
the block diagonal matrix whose $i$th diagonal block is the matrix $M_i$ is denoted by ${\rm diag} \{M_i\}_{i\in[1,p]}$. Given two sets $A$ and $B$, the {\em difference set} $A \setminus B$ is the set of elements of $A$ that are not in $B$. 
\smallskip
 
A {\em weighted, directed graph} (digraph) is   a triple ${\mathcal G} = ({\mathcal V}, {\mathcal E}, {\mathcal A})$, where ${\mathcal V} = \{1, \dots, N\}=[1,N]$ is the set of nodes, ${\mathcal E}\subset {\mathcal V}\times {\mathcal V}$  is the set of edges, and ${\mathcal A}\in \mathbb{R}^{N\times N}$ is the nonnegative, weighted {\em adjacency matrix} which satisfies $[{\mathcal A}]_{i,j}>  0$ if and only if $(j,i)\in {\mathcal E}$. 
We assume that $[{{\mathcal A}}]_{i,i} =0$ for every $i\in {\mathcal V}$.
The  {\em in-degree} of  node $i$ is $d_i \doteq \sum_{j=1}^N[{{\mathcal A}}]_{i,j}$, while the {\em in-degree matrix} $\mathcal{D}$ is the diagonal matrix defined as $\mathcal{D}\doteq {\rm diag}\{d_i\}_{i\in [1,N]}$. The {\em Laplacian}  associated with ${\mathcal G}$ is defined as ${\mathcal{L}}\doteq \mathcal{D}-\mathcal{A}$. 

The {\em spectrum} of a square matrix $A$, denoted by $\sigma(A)$, is the set of all its eigenvalues.
We let $r$ be the number of distinct eigenvalues of $A$, so that $\sigma(A) = \{\lambda_1, \dots, \lambda_r\}$.  The {\em algebraic multiplicity} and the {\em geometric multiplicity} of the eigenvalue $\lambda_\ell\in {\mathbb C}$  are denoted by $a_\ell$ and $g_\ell$, respectively. 
A matrix $A\in  {\mathbb C}^{n \times n}$ is in {\em Jordan form} \cite{HornJohnson} if it takes the form
\begin{equation}\label{Jform}
    A = {\rm diag} \{A^\ell\}_{\ell\in [1,r]},\end{equation}
where $A^\ell$ is the {\em Jordan block} of size $a_\ell \times a_\ell$ associated  with the eigenvalue $\lambda_\ell$, by this meaning that
\begin{equation}\label{Jblock}
A^\ell = {\rm diag} \{A^{\ell,h_\ell}\}_{h_\ell\in [1,g_\ell]}\in {\mathbb C}^{a_\ell \times a_\ell},\end{equation}
and $A^{\ell, h_\ell}$ is the $h_\ell$th {\em Jordan miniblock} associated with $\lambda_\ell$ and hence taking the form
\begin{equation}\label{Jminiblock}
A^{\ell, h_\ell} = \begin{bmatrix} \lambda_\ell & 1 & 0 &\dots & 0\cr   
0 & \lambda_\ell & 1 & \dots &0\cr
\vdots &\vdots &\ddots &\ddots& \vdots \cr
0 & 0 &0 & \ddots & 1\cr
0 & 0 &0 & \dots & \lambda_\ell\end{bmatrix}\in{\mathbb C
}^{d^{\ell, h_\ell} \times d^{\ell, h_\ell}}.\end{equation}
 \smallskip
 
\section{Problem set-up}\label{sec.pb.setup}
Consider a discrete-time linear time-invariant state-space model
\begin{align}\label{eq.sys}
    x(t+1)=Ax(t)+Bu(t),
    \qquad t\in\mathbb{Z_+},
\end{align}
where $x(t)\in\mathbb{R}^n$ is the state and $u(t)\in\mathbb{R}^m$ the control input. For simplicity, we assume that $A\in\mathbb{R}^{n\times n}$ is  {\em  in Jordan form} \eqref{Jform}-\eqref{Jblock}-\eqref{Jminiblock}, which implies, in particular, that all its $r$ distinct eigenvalues, $\lambda_1, \lambda_2, \dots, \lambda_r$,  are real. It entails no loss of generality  assuming that the eigenvalues have been sorted in such a way that $|\lambda_\ell|\ge1$  for $\ell\in [1,r_u]$, while $|\lambda_\ell|<1$  for $\ell\in [r_u+1,r]$. 
In this way, $r_u$ denotes the number of distinct unstable eigenvalues of $A$. 
\begin{remark}The assumption that $A$ is in Jordan form is not restrictive, since we can always reduce ourselves to that case by means of a change of basis.  The assumption that all the eigenvalues of $A$ are real is restrictive, but the extension to the general case of complex eigenvalues can be easily obtained by resorting to the {\em real Jordan form}, as in \cite{MitraSundaram2019} and \cite{GaoYang}. Both assumptions have been introduced only to simplify the notation.
\end{remark}
We assume that there is a network of $N$ sensors, connected through a directed graph $\mathcal{G}=({\mathcal V}, \mathcal{E},\mathcal{A})$, each of them taking an indirect measure of the state:
\begin{align}\label{eq.out}
    y_i(t)=C_ix(t),\quad i\in[1,N],
\end{align}
where $y_i(t)\in\mathbb{R}^{p_i}$. 
The objective is to design, for every node $i \in [1,N]$, a local observer that produces an estimate $\hat{x}_i(t)$ of the true state $x(t)$, such that the $i$th estimation error
\begin{align*}
    e_i(t)\doteq x(t)- \hat{x}_i(t),
\end{align*}
 asymptotically goes to zero, that is,
\begin{align*}
    \lim_{t \to +\infty} e_i(t) = \mathbb{0}, \qquad \forall i \in [1,N].
\end{align*}
Of course, a necessary condition for the previous problem to be solvable is that by making use of the output measurements of all the agents, it is possible to estimate the state of the system. Therefore, we introduce the following assumption.
\smallskip

\begin{assumption}\label{ass.jointdet}
\cite{GaoYang}
The $N$ agents are {\em jointly detectable}, i.e.,
$$\left(A, \begin{bmatrix} C_1 \cr \vdots \cr C_N \end{bmatrix}\right)
\ \mbox{is a detectable pair.}$$
\end{assumption}
\smallskip

\noindent According to the  block-partition of $A$, we divide the state as 
\begin{align*}
    x(t)={\rm col}\{x^\ell(t)\}_{\ell\in[1,r]},\quad x^\ell(t)={\rm col}\{x^{\ell,h_\ell}(t)\}_{h_\ell\in[1,g_\ell]},
\end{align*}
and the output matrices $C_i$, $i\in[1,N],$ as
$$C_i = \begin{bmatrix} C_i^1 & \dots &C_i^r\end{bmatrix}\in {\mathbb R}^{p_i \times n},$$
where $C_i^\ell$, $\ell\in[1,r],$ are in turn partitioned as
$$C_i^\ell = \begin{bmatrix} C_i^{\ell,1} & \dots& C_i^{\ell,g_\ell}\end{bmatrix}\in {\mathbb R}^{p_i \times a_\ell},$$
and each   block $C_i^{\ell, h_\ell}$ has size $p_i\times d^{\ell, h_\ell}$. We now want to act on the matrix pairs $(A,C_i), i\in [1,N],$ in order to sort, for every $i$, the entries of the vector $x(t)$ according to the detectability properties of the pair $(A,C_i)$. The procedure is close to the one proposed in \cite{GaoYang}, but introduces some  changes that not only lead to a simpler notation but also pave the way to simpler solvability conditions. 
\medskip

{\bf Step 1.}\ For every $i\in [1,N]$ and $\ell\in [1,r_u]$,  we partition the set of indices $h_\ell \in [1, g_\ell]$ into three disjoint sets, based on the properties of the block $C_i^{\ell,h_\ell}$ in $C_i$ (corresponding to the Jordan miniblock $A^{\ell,h_\ell}$):
\begin{subequations}
\label{eq:defGil}
    \begin{align}
    \mathbb{G}_{i,1}^{\ell}&\doteq\{h_\ell\in[1,g_\ell]:C_i^{\ell,h_\ell}=\mathbb{0}\},\\
    \mathbb{G}_{i,2}^{\ell}&\doteq\{h_\ell\in[1,g_\ell]:C_i^{\ell,h_\ell}\mathbb{e}_1=\mathbb{0}\  \wedge\ C_i^{\ell,h_\ell}\ne\mathbb{0} \},\\
    \mathbb{G}_{i,3}^{\ell}&\doteq\{h_\ell\in[1,g_\ell]:C_i^{\ell,h_\ell}\mathbb{e}_1\ne\mathbb{0}\}.
\end{align}
\end{subequations}
Clearly, $\mathbb{G}_{i,1}^{\ell} \cup \mathbb{G}_{i,2}^{\ell} \cup \mathbb{G}_{i,3}^{\ell} = [1, g_\ell].$ On the other hand, for every pair $(\ell, h_\ell)$, with $\ell\in [1,r_u]$ and  $h_\ell\in [1, g_\ell]$, 
one can  introduce the index sets
\begin{subequations}
    \begin{align}
\mathcal{V}_{1}^{\ell,h_\ell}&\doteq\{i:h_\ell\in\mathbb{G}_{i,1}^{\ell}\},\\
\mathcal{V}_{2}^{\ell,h_\ell}&\doteq\{i:h_\ell\in\mathbb{G}_{i,2}^{\ell}\},\\
\mathcal{V}_{3}^{\ell,h_\ell}&\doteq\{i:h_\ell\in\mathbb{G}_{i,3}^{\ell}\}.
\end{align}
\end{subequations}
Note that the three sets are disjoint, and
$\mathcal{V}_{1}^{\ell,h_\ell} \cup \mathcal{V}_{2}^{\ell,h_\ell} \cup \mathcal{V}_{3}^{\ell,h_\ell} =[1,N].$
\begin{proposition}\label{prop.v3Nonempty}
Under Assumption \ref{ass.jointdet}, for every pair $(\ell, h_\ell)$, with $\ell\in [1,r_u]$ and  $h_\ell\in [1, g_\ell]$, the set ${\mathcal V}_3^{\ell,h_\ell}$ is not empty, or, equivalently,
there exists $i\in [1,N]$ such that $h_\ell\in {\mathbb G}_{i,3}^\ell.$
\end{proposition}
\begin{proof}
    By Assumption \ref{ass.jointdet} and the hypothesis that $A$ is  Jordan form, it follows from the PBH detectability test that for every $\ell\in [1,r_u]$ and every $h_\ell\in[1,g_\ell]$ there exists $i\in [1,N]$ such that $C_{i}^{\ell,h_\ell} {\mathbb e}_1 \ne {\mathbb 0}.$
    This implies that $h_\ell \in {\mathbb G}_{i,3}^\ell$.
\end{proof}
\begin{remark} It is worth observing that for every $h_\ell\in {\mathbb G}_{i,3}^\ell$ the pair $(A^{\ell, h_\ell},C_i^{\ell, h_\ell})$ is observable. So,
Proposition \ref{prop.v3Nonempty} states  that, under Assumption \ref{ass.jointdet},  for  every $\ell\in [1,r_u]$ and every $h_\ell\in [1, g_\ell]$, there exists at least one agent $i\in \mathcal{V}$ such that the pair $(A^{\ell,h_\ell},C_i^{\ell,h_\ell})\in {\mathbb R}^{d^{\ell, h_\ell}\times d^{\ell, h_\ell}} \times {\mathbb R}^{p_i \times d^{\ell, h_\ell}}$ is observable.  
\end{remark}

Based on the definitions \eqref{eq:defGil} and the fact that $\mathbb{G}_{i,1}^{\ell} \cup \mathbb{G}_{i,2}^{\ell} \cup \mathbb{G}_{i,3}^{\ell} = [1, g_\ell]$, we can   claim that for every $i\in[1,N]$ and $\ell\in [1,r_u]$, 
there exists a permutation matrix $P^{\ell}_i$  such that 
\begin{align*}
 ({P^{\ell}_i})^\top A^{\ell}P^{\ell}_i&=\begin{bmatrix}
    A_{i,1}^{\ell} &\vline&\mathbb{0}&\mathbb{0}\\
       \hline
    \mathbb{0}&\vline&A_{i,2}^{\ell} &\mathbb{0}\\
    \mathbb{0}&\vline&\mathbb{0}&A_{i,3}^{\ell}
\end{bmatrix},
  ({P^{\ell}_i})^\top B^\ell=\begin{bmatrix}
    B_{i,1}^{\ell}\\
    \hline
    B_{i,2}^{\ell}\\
    B_{i,3}^{\ell}
\end{bmatrix}\\
C_iP^{\ell}_i&=\begin{bmatrix}
    \mathbb{0}&\vline&C_{i,2}^{\ell}&C_{i,3}^{\ell}
\end{bmatrix},\\
 ({P^{\ell}_i})^\top x^\ell(t)&=\begin{bmatrix}
    x_{i,1}^{\ell}(t)\\
    \hline
    x_{i,2}^{\ell}(t)\\
    x_{i,3}^{\ell}(t)
\end{bmatrix}, 
\end{align*}
where the block diagonal matrices $A_{i,1}^\ell, A_{i,2}^\ell$ and $A_{i,3}^\ell$ group all Jordan miniblocks corresponding to indices $h_\ell$ belonging to ${\mathbb G}_{i,1}^{\ell}, {\mathbb G}_{i,2}^{\ell}$ and ${\mathbb G}_{i,3}^{\ell}$, respectively, i.e.,
\begin{align*}
    A_{i,1}^\ell&\doteq{\rm diag}\{A^{\ell,h_\ell}\}_{h_\ell\in\mathbb{G}_{i,1}^{\ell}},\\
    A_{i,2}^\ell&\doteq{\rm diag}\{A^{\ell,h_\ell}\}_{h_\ell\in\mathbb{G}_{i,2}^{\ell}},\\
    A_{i,3}^\ell&\doteq{\rm diag}\{A^{\ell,h_\ell}\}_{h_\ell\in\mathbb{G}_{i,3}^{\ell}}.
\end{align*}
The blocks of $(P_i^\ell)^\top B^{\ell}$ and $C_iP_i^\ell$ are defined accordingly.
In particular,  the matrix $C_{i,2}^\ell$ (respectively, $C_{i,3}^\ell$)  consists of the blocks $C_i^{\ell,h_\ell}$ of the matrix $C_i$ corresponding to the Jordan miniblocks $A^{\ell,h_\ell}$ with $h_\ell\in \mathbb{G}_{i,2}^{\ell}$ (respectively, $h_\ell\in \mathbb{G}_{i,3}^{\ell}$).
\medskip

Note that if $h_\ell \in \mathbb{G}_{i,2}^{\ell}$ and we let $t_i^{\ell,h_\ell}\in [2, d^{\ell, h_\ell}]$ denote the index of the first nonzero column of $C_i^{\ell,h_\ell},$
then $C_i^{\ell,h_\ell}$ takes the form
$$C_i^{\ell,h_\ell} = \begin{bmatrix} {\mathbb 0}_{p_i\times (t_i^{\ell,h_\ell}-1)} & C_{i,o}^{\ell,h_\ell}\end{bmatrix},$$
and the first column of $C_{i,o}^{\ell,h_\ell}$ is nonzero.  Moreover, the corresponding miniblock $A^{\ell, h_\ell}$ can be block partitioned as
$$A^{\ell,h_\ell} = \begin{bmatrix} A^{\ell,h_\ell}_{i,u} &A^{\ell,h_\ell}_{i,*}\cr\mathbb{0} & A^{\ell,h_\ell}_{i,o}\end{bmatrix},$$
where $A^{\ell,h_\ell}_{i,u}$ and $A^{\ell,h_\ell}_{i,o}$ 
are described as in \eqref{Jminiblock} (and hence 
have, in turn, the structure of two Jordan miniblocks) and have  sizes $t_i^{\ell,h_\ell}-1$ and $d^{\ell, h_\ell} - (t_i^{\ell,h_\ell}-1)$, respectively. 
It is easy to see that the pair $(A^{\ell,h_\ell}_{i,o},C_{i,o}^{\ell,h_\ell})$ is observable.\\
To improve the paper readability, we have grouped the main symbols
in the following table.\\

{\centering
\begin{tabular}{|c|p{6.6cm}|}
\hline
\textbf{Symbol} & \textbf{Meaning} \\
\hline
$\lambda_\ell$ &Generic eigenvalue of $A$\\
$\alpha_\ell$ & Algebraic multiplicity of $\lambda_\ell$ \\
$g_\ell$ & Geometric multiplicity of $\lambda_\ell$\\
$A^\ell$ & Jordan block associated with $\lambda_\ell$ \\
$A^{\ell,h_\ell}$ & $h_\ell$th Jordan miniblock associated with $\lambda_\ell$ \\
 $h_\ell$ & Index of a generic Jordan miniblock associated with $\lambda_\ell$ \\
$d^{\ell, h_\ell}$ & Dimension of the $h_\ell$th Jordan miniblock  associated with $\lambda_\ell$ \\
$t_i^{\ell,h_\ell}$ & Index of the first nonzero column of the block $C_i^{\ell,h_\ell}$, for $h_\ell\in {\mathbb G}_{i,2}^\ell$\\
\hline
\end{tabular}}
\medskip

{\bf Step 2.} We now focus on the pair $(A_{i,2}^\ell,C_{i,2}^\ell)$ and permute the entries in such a way that all the entries corresponding to the first $t_i^{\ell,h_\ell} -1$ entries of each Jordan miniblock $A^{\ell, h_\ell}, h_\ell\in {\mathbb G}_{i,2}^\ell$ come first, while the entries corresponding to the last $d^{\ell, h_\ell} - (t_i^{\ell,h_\ell}-1)$ entries for each such miniblock are at the end. This amounts to considering the permuted matrices
\begin{align*}
&{(P_{i,2}^\ell)}^\top A_{i,2}^\ell{P_{i,2}^\ell}=\begin{bmatrix}
A_{i,2u}^{\ell}&\vline&A_{i,2*}^{\ell}\\
    \hline
 \mathbb{0} &\vline& A_{i,2o}^{\ell}
\end{bmatrix}, \ {(P_{i,2}^\ell)}^\top x_{i,2}^\ell=\begin{bmatrix}
    x_{i,2u}^{\ell}\\
    \hline
    x_{i,2o}^{\ell}
\end{bmatrix},\\
&{(P_{i,2}^\ell)}^\top  B^\ell=\begin{bmatrix}
    B_{i,2u}^{\ell}\\
    \hline
    B_{i,2o}^{\ell}
\end{bmatrix},\quad C_{i,2}^\ell{P_{i,2}^\ell}=\begin{bmatrix}
    \mathbb{0}&\vline&C_{i,2o}^{\ell}
\end{bmatrix},
\end{align*}
where 
\begin{align*}
A_{i,2u}^\ell &\doteq {\rm diag} \{A^{\ell,h_\ell}_{i,u} \}_{h_\ell \in {\mathbb G}_{i,2}^\ell},\\
A_{i,2o}^\ell &\doteq {\rm diag} \{A^{\ell,h_\ell}_{i,o} \}_{h_\ell \in {\mathbb G}_{i,2}^\ell}, \\
C_{i,2o}^{\ell}&\doteq {\rm row} \{C^{\ell,h_\ell}_{i,o} \}_{h_\ell \in {\mathbb G}_{i,2}^\ell},
\end{align*}
namely $A_{i,2u}^\ell, A_{i,2o}^\ell$ and $C_{i,2o}^{\ell}$ consist of all blocks $A_{i,u}^{\ell,h_\ell}, A_{i,o}^{\ell,h_\ell}$ and
$C_{i,o}^{\ell,h_\ell}$, respectively,  corresponding to the indices $h_\ell\in \mathbb{G}_{i,2}^{\ell}$. The other blocks are defined accordingly.
\smallskip

{\bf Step 3.}\ We now put together the outcomes of the two previous steps and introduce a final permutation that groups  analogous blocks corresponding to distinct (unstable) eigenvalues. We define  
\begin{align*}
 A_{i,1}&\doteq{\rm diag}\{A_{i,1}^{\ell}\}_{\ell\in [1,r_u]},\quad 
 A_{i,3}\doteq{\rm diag}\{A_{i,3}^{\ell}\}_{\ell\in [1,r_u]},  \\
A_{i,2u}&\doteq{\rm diag}\{A_{i,2u}^{\ell}\}_{\ell\in [1,r_u]},\quad  
A_{i,2*}\doteq{\rm diag}\{A_{i,2*}^{\ell}\}_{\ell\in [1,r_u]},\\ 
A_{i,2o}&\doteq{\rm diag}\{A_{i,2o}^{\ell}\}_{\ell\in [1,r_u]},
\\
 B_{i,1}&\doteq{\rm col}\{{B_{i,1}^{\ell}}\}_{\ell\in [1,r_u]},\quad 
  B_{i,3}\doteq{\rm col}\{{B_{i,3}^{\ell}}\}_{\ell\in [1,r_u]},\\
 B_{i,2u}&\doteq{\rm col}\{{B_{i,2u}^{\ell}}\}_{\ell\in [1,r_u]},\quad 
B_{i,2o}\doteq{\rm col}\{{B_{i,2o}^{\ell}}\}_{\ell\in [1,r_u]},\\
C_{i,2o}&\doteq{\rm row}\{{C_{i,2o}^{\ell}}\}_{\ell\in [1,r_u]},\quad
C_{i,3}\doteq{\rm row}\{{C_{i,3}^{\ell}}\}_{\ell\in [1,r_u]}, \\
 x_{i,1}(t)&\doteq{\rm col}\{{x_{i,1}^{\ell}(t)}\}_{\ell\in [1,r_u]},\
x_{i,3}(t)\doteq{\rm col}\{{x_{i,3}^{\ell}(t)}\}_{\ell\in [1,r_u]},\\
x_{i,2u}(t)&\doteq{\rm col}\{{x_{i,2u}^{\ell}(t)}\}_{\ell\in [1,r_u]},
x_{i,2o}(t)\doteq{\rm col}\{{x_{i,2o}^{\ell}(t)}\}_{\ell\in [1,r_u]}
\end{align*}
We also group together Jordan blocks corresponding to stable eigenvalues:
\begin{align*}
A_{s}&\doteq{\rm diag}\{A^{\ell}\}_{\ell\in[r_u+1,r]},\quad B_{s}\doteq{\rm diag}\{B^{\ell}\}_{\ell\in[r_u+1,r]},\\
 C_{s}& \doteq{\rm row}\{C^{\ell}\}_{\ell\in[r_u+1,r]}, \quad x_{i,s}(t)\doteq{\rm col}\{{x_{i,s}^{\ell}(t)}\}_{\ell\in [1,r_u]}.
\end{align*}
As a result of all previous re-orderings and definitions,  we can claim that for every $i\in [1,N]$
there exists a permutation matrix $Q_i\in {\mathbb R}^{n \times n}$ such that 
\begin{align}
{Q_i}^\top AQ_i&=\begin{bmatrix}
    F_{i,u}&\vline& F_{i,*}\\
    \hline
    \mathbb{0}&\vline& F_{i,d}
\end{bmatrix} \nonumber \\
&\doteq \begin{bmatrix}
    A_{i,1}&\mathbb{0} &\vline&\mathbb{0}&\mathbb{0}&\mathbb{0}\\
    \mathbb{0}&A_{i,2u}&\vline&A_{i,2*} &\mathbb{0}&\mathbb{0}\\
    \hline
    \mathbb{0} &\mathbb{0}&\vline&A_{i,2o}&\mathbb{0}&\mathbb{0} \\
    \mathbb{0}&  \mathbb{0}&\vline&\mathbb{0}&A_{i,3}&\mathbb{0}\\
        \mathbb{0}&  \mathbb{0}&\vline&\mathbb{0}&\mathbb{0}&A_s
\end{bmatrix} \nonumber\\
C_iQ_i&=\begin{bmatrix}
    \mathbb{0}&\vline&H_{i,d}
\end{bmatrix} 
\doteq\begin{bmatrix}
    \mathbb{0}&\mathbb{0}&\vline&C_{i,2o}&C_{i,3}&C_{s}
\end{bmatrix}, \nonumber\\
{Q_i}^\top x(t)&=\begin{bmatrix}
    z_{i,u}(t)\\
    \hline
    z_{i,d}(t)
\end{bmatrix}
\doteq\begin{bmatrix}
    x_{i,1}(t)\\
    x_{i,2u}(t)\\
    \hline
    x_{i,2o}(t)\\
    x_{i,3}(t)\\
    x_{i,s}(t)
\end{bmatrix}, \label{eq.zdzu}\\ {Q_i}^\top B&=\begin{bmatrix}
    G_{i,u}\\
    \hline
    G_{i,d}
\end{bmatrix} 
\doteq\begin{bmatrix}
    B_{i,1}\\
    B_{i,2u}\\
    \hline
    B_{i,2o}\\
    B_{i,3}\\
    B_{s}
\end{bmatrix}.  \nonumber
\end{align}

From now on, we will rely on the following assumption\footnote{Note that this assumption is also used in \cite{GaoYang}, but it remains hidden. Differently, the matrix governing the estimation error dynamics of the Luenberger observer proposed in (21) would not be Schur stable.}.
\begin{assumption}\label{ass.obs}
    For every $i\in[1,N]$ and $\ell\in [1,r_u]$, the vectors 
    \begin{equation}
    \{C_i^{\ell,h_\ell}\mathbb{e}_{t_i^{\ell,h_\ell}}
    \}_{h_\ell\in\mathbb{G}_{i,2}^{\ell}} \cup \{C_i^{\ell,h_\ell}\mathbb{e}_{1}\}_{h_\ell\in\ \mathbb{G}_{i,3}^{\ell}} 
     \end{equation} 
    are a linearly independent set.  
    \end{assumption}

It is worth remarking that if Assumption \ref{ass.obs} holds, then it follows from the observability of the pairs 
$(A^{\ell,h_\ell},C_i^{\ell,h_\ell})$  for all $h_\ell\in\mathbb{G}_{i,3}^{\ell}$, and of the pairs 
$(A_{i,2o}^\ell,C_{i,2o}^\ell)$ for all $h_\ell\in \mathbb{G}_{i,2}^{\ell}$, 
 that for every $i\in[1,N]$ and $\ell\in [1,r_u]$, the pair 
$$\left(
    \begin{bmatrix}
        A_{i,2o}^{\ell}&\mathbb{0}\\
        \mathbb{0}&A_{i,3}^{\ell}
    \end{bmatrix},\begin{bmatrix}
    C_{i,2o}^{\ell}& C_{i,3}^{\ell}
\end{bmatrix}\right)$$ 
    is observable. Therefore for every $i\in [1,N]$ the pair $(F_{i,d}, H_{i,d})$ (that includes also the stable eigenvalues of $A$, namely the block $A_s$) is detectable.

\section{Problem Solution}\label{sec.pb.solution}

{\bf Preliminary estimate of $z_{i,d}(t)$:}\ By relying on the fact that for each $i\in [1,N]$ the pair 
$(F_{i,d}, H_{i,d})$ is detectable, we can
 design a Luenberger observer that allows agent $i$ to estimate the part of the vector $x(t)$ that is detectable for it, namely $z_{i,d}(t)$ (see \eqref{eq.zdzu}), as follows:
\begin{align}\label{eq.luenberger}
\hat z_{i,d}(t+1)=&F_{i,d}\hat z_{i,d}(t)+G_{i,d}u(t)\\
    &+L_{i,d}(y_i(t)-H_{i,d}\hat z_{i,d}(t)),\nonumber
\end{align}
where $L_{id}$ is chosen so that $F_{i,d} - L_{i,d} H_{i,d}$ is Schur stable. 
If we define the estimation error of the part of the state that is detectable for agent $i$ as $\eta_{i,d}(t) \doteq z_{i,d}(t)-\hat z_{i,d}(t)$, then the estimation error dynamics are given by
\begin{align}
     \eta_{i,d}(t+1)=(F_{i,d}-L_{i,d}H_{i,d}) \eta_{i,d}(t),
\end{align} 
and $\eta_{i,d}(t)$ asymptotically goes to zero, for every initial condition $\eta_{i,d}(0)$.
\smallskip

{\bf Alternative decomposition:}\ We now  introduce another permutation matrix, alternative to $Q_i$. After having introduced the sets ${\mathbb G}_{i,k}^\ell, k\in[1,3],$ and the block diagonal matrices $A_{i,k}^\ell, k\in [1,3],$ for every $i\in [1,N]$ and every $
\ell\in [1,r_u],$ as in Step 1, we let $R_i$ denote a permutation matrix that orders the unstable miniblocks of the matrix $A$ according to the sets $\mathbb{G}_{i,1}^{\ell}$, $\mathbb{G}_{i,2}^{\ell}$, and $\mathbb{G}_{i,3}^{\ell}$. 
This amounts to saying that\footnote{Compared to the permutation $Q_i$ we adopted in Step 2, we have simply ordered the blocks $A^{\ell, h_\ell}, h_\ell\in {\mathbb G}_{i,2}^\ell$, without splitting them on the basis of the (zero columns of the) corresponding block $C_i^{\ell, h_\ell}$.}
\begin{align}
{R_i}^\top AR_i&=\begin{bmatrix}
    A_{i,u}&\vline& {\mathbb 0}\\
    \hline
    \mathbb{0}&\vline& A_{i,d}
\end{bmatrix} = \begin{bmatrix}
    A_{i,1}&\mathbb{0} &\vline&\mathbb{0}&\mathbb{0}\\
    \mathbb{0}&A_{i,2}&\vline& \mathbb{0} &\mathbb{0}\\
    \hline
    \mathbb{0}&  \mathbb{0}&\vline& A_{i,3}&\mathbb{0}\\
        \mathbb{0}&  \mathbb{0}&\vline&\mathbb{0}&A_s
\end{bmatrix}, \nonumber\\
C_iR_i&= 
 \begin{bmatrix}
    \mathbb{0}&C_{i,2}&\vline& C_{i,3}&C_{s}
\end{bmatrix}, \nonumber\\
{R_i}^\top B &=
\begin{bmatrix}
     B_{i,u}\\
     \hline
     B_{i,d}
 \end{bmatrix} =
\begin{bmatrix}
    B_{i,1}\\
    B_{i,2}\\
    \hline
    B_{i,3}\\
    B_{s}
\end{bmatrix},  \nonumber
\end{align}
where $A_{i,2}\doteq{\rm diag}\{A_{i,2}^{\ell}\}_{\ell\in [1,r_u]}$, $ B_{i,2}\doteq{\rm col}\{{B_{i,2}^{\ell}}\}_{\ell\in [1,r_u]}$, and $ C_{i,2}\doteq{\rm row}\{{C_{i,2}^{\ell}}\}_{\ell\in [1,r_u]}$.
This corresponds to splitting the (permuted) state vector as follows:
\begin{align*}
{R_i}^\top x(t)=\begin{bmatrix}
    x_{i,u}(t)\\
    \hline
    x_{i,d}(t)
\end{bmatrix}\doteq\begin{bmatrix}
    x_{i,1}(t)\\
    x_{i,2}(t)\\
    \hline
    x_{i,3}(t)\\
    x_{i,s}(t)
\end{bmatrix},\end{align*}
where  $x_{i,2}(t)\doteq{\rm col}\{{x_{i,2}^{\ell}(t)}\}_{\ell\in [1,r_u]}$. \\
Agent $i$ will separately estimate $x_{i,d}(t)$ and $x_{i,u}(t)$ using two different observers.
The estimate $\hat x_{i,d}(t)$ will be obtained from the estimate  $\hat z_{i,d}(t)$ derived in  \eqref{eq.luenberger},
while the estimate  $\hat x_{i,u}(t)$ will be obtained using the estimates   of $x_{i,u}(t)$ provided by agent $i$th neighbors. Note that each entry of the estimates provided by the neighbors may be part of either the detectable or the undetectable part of the neighbor state vectors, and therefore for every $j\in {\mathcal N}_i$, where ${\mathcal N}_i \doteq \{j\in [1,N]: (j,i) \in {\mathcal E}\}$ is the set of indices of the neighbors of $i$,  such estimated entries may be part   either of $\hat x_{j,d}(t)$ or of $\hat x_{j,u}(t)$.

The detailed derivation is provided in the following. 
 The final estimate that agent $i$ provides of the state $x(t)$ will be
\begin{align*}
\hat x_i(t) = R_i \begin{bmatrix}
    \hat x_{i,u}(t)\\
    \hline
    \hat x_{i,d}(t)
\end{bmatrix}.\end{align*}
\smallskip

{\bf Estimate of $x_{i,d}(t)$:}\
The
estimate $x_{i,d}(t)$ provided  by agent $i$ relies on the Luenberger observer \eqref{eq.luenberger} and is described  by the following state-space  model:
\begin{subequations}\label{eq.lunberger.directed}
\begin{align}
    \hat z_{i,d}(t+1)=&F_{i,d}\hat z_{i,d}(t)+G_{i,d}u(t)\\
    &+L_{i,d}(y_i(t)-H_{i,d}\hat z_{i,d}(t)),\nonumber\\
    \hat x_{i,d}(t) =& \begin{bmatrix}
        \mathbb{0}&I
    \end{bmatrix} \hat z_{i,d}(t),\label{eq.lunberger.directed.b}
\end{align}
\end{subequations}
where we exploited the relationship between $x_{i,d}(t)$ and $z_{i,d}(t)$, namely:
$$x_{i,d}(t)=\begin{bmatrix}
        x_{i,3}(t)\\
         x_{i,s}(t)
    \end{bmatrix}=\begin{bmatrix}
        \mathbb{0}&I
    \end{bmatrix}
    \begin{bmatrix}
     x_{i,2o}(t)\\
     x_{i,3}(t)\\
     x_{i,s}(t)
\end{bmatrix} = \begin{bmatrix}
        \mathbb{0}&I
    \end{bmatrix}  z_{i,d}(t).$$
As noted previously, the detectability of the pair $(F_{i,d}, H_{i,d})$ ensures that there exists $L_{i,d}$ such that $F_{i,d}- L_{i,d} H_{i,d}$ is Schur stable, and hence $\eta_{i,d}(t)= z_{i,d}(t) - \hat z_{i,d}(t)$ asymptotically converges to $0$.  Therefore
$$e_{i,d}(t) \doteq x_{i,d}(t) - \hat x_{i,d}(t)= \begin{bmatrix}
        \mathbb{0}&I
    \end{bmatrix}  \eta_{i,d}(t)$$
asymptotically converges to $0$, in turn.
\smallskip

{\bf Estimate of $x_{i,u}(t)$:}\
We design an observer of the type
\begin{align}
   \hat  x_{i,u}(t+1)=&A_{i,u}\hat x_{i,u}(t)+B_{i,u}u(t) \label{eq.xiuhat}\\
&+K_iA_{i,u}\sum_{j\in\mathcal{N}_i}[\mathcal{A}]_{i,j}\begin{bmatrix}
        I&\mathbb{0}
    \end{bmatrix}R_i^\top(\hat x_j(t)-\hat x_i(t))
    \nonumber
\end{align}
 where ${\mathcal A}$ is the adjacency matrix of the communication digraph ${\mathcal G}$ and $K_i$ is a block-diagonal matrix to be designed in the following form:
\begin{equation*}
  {\small  K_i \doteq \begin{bmatrix}{\rm diag} \left\{ K_{i,1}^\ell\right\}_{\ell\in [1,r_u]} & \mathbb{0} \cr \mathbb{0}&{\rm diag} \left\{K_{i,2}^\ell\right\}_{\ell\in [1,r_u]  }\end{bmatrix}
  }
\end{equation*}
where
\begin{align*}
K_{i,1}^\ell=&\ {\rm diag} \left\{ k^{\ell,h_\ell} I_{d^{\ell, h_\ell}}\right\}_{ h_\ell \in {\mathbb G}_{i,1}^\ell}\\
K_{i,2}^\ell=&\ {\rm diag} \left\{ k^{\ell,h_\ell} I_{d^{\ell, h_\ell}}\right\}_{ h_\ell \in {\mathbb G}_{i,2}^\ell}.
\end{align*}
Note that this amounts to choosing a scalar gain $k^{\ell,h_\ell}$ for each Jordan miniblock $A^{\ell, h_\ell}$, corresponding to some $\ell\in [1,r_u]$ and $h_\ell\in {\mathbb G}_{i,1}^\ell \cup {\mathbb G}_{i,2}^\ell.$ Recall that ${\mathbb G}_{i,1}^\ell \cap {\mathbb G}_{i,2}^\ell=\emptyset$.\\ In order to study the dynamics of the estimation error $e_{i,u} (t) \doteq x_{i,u}(t)-\hat x_{i,u}(t)$, we exploit the block diagonal structure of the matrices. Indeed, it is easy to deduce from \eqref{eq.xiuhat} that for every $\ell\in [1,r_u]$ and every $h_\ell \in {\mathbb G}_{i,1}^\ell \cup {\mathbb G}_{i,2}^\ell$, the  portion of the vector $\hat x_{i,u}(t)$ associated with the miniblock $A^{\ell,h_\ell}$, i.e., $\hat x_{i}^{\ell,h_\ell}(t)$, updates  according to
\begin{align}\label{eq.estimate.undetectble.directed}
        \hat x_{i}^{\ell,h_\ell}&(t+1)=
        A^{\ell,h_\ell}
        \hat x_{i}^{\ell,h_\ell}(t)
+B^{\ell,h_\ell} u(t)\\&+k^{\ell,h_\ell}
        A^{\ell,h_\ell}\sum_{j\in\mathcal{N}_i}[\mathcal{A}]_{i,j}(\hat x_j^{\ell,h_\ell}(t)-\hat x_i^{\ell,h_\ell}(t)).\nonumber
\end{align}
It is worth remarking that $\hat x_j^{\ell,h_\ell}(t)$ is the estimate of $x^{\ell,h_\ell}(t)$ provided by agent $j\in {\mathcal N}_i$, and it is either part of $\hat x_{j,d}(t)$ or of $\hat x_{j,u}(t)$, depending on whether $h_\ell \in {\mathbb G}_{j,3}^\ell$ or $h_\ell \in {\mathbb G}_{j,1}^\ell \cup {\mathbb G}_{j,2}^\ell$.\\
Since the dynamics of the portion of the state $x^{\ell,h_\ell}(t)$ associated with the miniblock $A^{\ell,h_\ell}$ is
\begin{align}\label{eq.miniblock}
         x^{\ell,h_\ell}(t+1)=A^{\ell,h_\ell}x^{\ell,h_\ell}(t)+B^{\ell,h_\ell} u(t),
\end{align}
if we define the error that agent $i$ makes in estimating $x^{\ell,h_\ell}(t)$
as $e_i^{\ell,h_\ell}(t) \doteq x^{\ell,h_\ell}(t)-\hat x_i^{\ell,h_\ell}(t)$, then the dynamics of such estimation error is
\begin{align}
     e_{i}^{\ell,h_\ell}&(t+1)=A^{\ell,h_\ell}
         e_{i}^{\ell,h_\ell}(t)\\
         &- k^{\ell,h_\ell}
        A^{\ell,h_\ell}\sum_{j\in \mathcal{ N}_{i}}[\mathcal{A}]_{i,j}(e_i^{\ell,h_\ell}(t)-e_j^{\ell,h_\ell}(t)),
        \nonumber
\end{align}
and by the definition of a Laplacian matrix $\mathcal{L}$, we can rewrite the previous equation as 
\begin{align}\label{eq.est.err.lap.dir}
     e_{i}^{\ell,h_\ell}(t+1)=&(1 - k^{\ell,h_\ell}[\mathcal{L}]_{i,i})A^{\ell,h_\ell}
         e_{i}^{\ell,h_\ell}(t)\\
         &- k^{\ell,h_\ell}
        A^{\ell,h_\ell}\sum_{j\in \mathcal{ N}_{i}}[\mathcal{L}]_{i,j}e_j^{\ell,h_\ell}(t).\nonumber
\end{align}
The dynamics of the error $e_{i}^{\ell,h_\ell}(t)$ of the $i$th agent is influenced by the estimation errors $e_j^{\ell,h_\ell}(t)$ of its neighbors ($j\in {\mathcal N}_i$), which can be split into two categories: those for which the $h_\ell$th miniblock is completely observable, namely $h_\ell\in\mathbb{G}_{j,3}^\ell$ (equivalently, $j \in \mathcal{V}_{3}^{\ell,h_\ell}$), and the others. 

Therefore,  the sum in \eqref{eq.est.err.lap.dir} can be rewritten as\footnote{To lighten the notation we have omitted to specify that $j$ belongs to ${\mathcal N}_i$, but this constraint is implicit in the fact that $[{\mathcal L}]_{i,j}=0$ if $j\not\in {\mathcal N}_i$.}
        \begin{align}\label{eilh}
        e_{i}^{\ell,h_\ell}(t+1)&=(1 -  k^{\ell,h_\ell}[\mathcal{L}]_{i,i})
         A^{\ell,h_\ell}e_{i}^{\ell,h_\ell}(t)\\
         &- k^{\ell,h_\ell}
        A^{\ell,h_\ell}\sum_{j\in (\mathcal{V}_{1}^{\ell,h_\ell} \cup \mathcal{V}_{2}^{\ell,h_\ell})}[\mathcal{L}]_{i,j}e_j^{\ell,h_\ell}(t) \nonumber
        \\&- k^{\ell,h_\ell}
        A^{\ell,h_\ell}\sum_{j\in \mathcal{V}_{3}^{\ell,h_\ell}}[\mathcal{L}]_{i,j}e_j^{\ell,h_\ell}(t).\nonumber
\end{align}
The last step is to design $k^{\ell,h_\ell}$ for each $(\ell,h_\ell)\in [1,r_u]\times [1, g_\ell]$, such that
\begin{align}
\label{eq.tuttinulli}
    \lim_{t\to+\infty}e_{i}^{\ell,h_\ell}(t)=\mathbb{0},\quad \forall\ i\in[1,N].
\end{align}
We observe that if we define
\begin{align*}
e_u(t) \doteq {\rm col}\{e_{i,u}(t)\}_{i\in [1,N]}, \\ 
e_d(t) \doteq {\rm col}\{e_{i,d}(t)\}_{i\in [1,N]}, \\
\eta_d(t) \doteq {\rm col}\{\eta_{i,d}(t)\}_{i\in [1,N]}, 
\end{align*}
then the dynamics of the estimation error 
$$\begin{bmatrix} e_u(t)\cr e_d(t)\end{bmatrix}$$
updates according to an equation of the following form:
\begin{align}
\label{eq.triang}
\begin{bmatrix} e_u(t+1)\cr \eta_d(t+1)\end{bmatrix}
&= \begin{bmatrix} \Phi_u & \Phi_*S \cr {\mathbb 0} & \Phi_d\end{bmatrix}
\begin{bmatrix} e_u(t)\cr \eta_d(t)\end{bmatrix}\\
\begin{bmatrix}
    e_u(t)\\
    e_d(t)
\end{bmatrix}&=\begin{bmatrix}
    I&\mathbb{0}\\ \mathbb{0}&S
\end{bmatrix}\begin{bmatrix} e_u(t)\cr \eta_d(t)\end{bmatrix}.
\end{align}
where $\Phi_d \doteq {\rm diag} \{ A_{i,d} - L_{i,d}C_{i,d}\}_{i\in [1,N]}$ is Schur,   $S$ is a selection matrix\footnote{It is actually a block diagonal matrix whose diagonal blocks are the matrices $\begin{bmatrix} {\mathbb 0} & I\end{bmatrix}$ appearing in \eqref{eq.lunberger.directed.b}.}, $\Phi_*$ and $\Phi_u$ can be obtained by making use of \eqref{eilh}. Note that the expressions of $S$ and $\Phi_*$ are irrelevant for the subsequent discussion. We will come later to the role of $\Phi_u$.\\
The convergence to zero of $\eta_{d}(t)$ and consequently of $e_{d}(t)$  (and, hence, of  every $e_{i,d}(t)$) has already been discussed. On the other hand, due to the block triangular structure of the matrix in \eqref{eq.triang}, 
condition \eqref{eq.tuttinulli} holds 
for every $\ell\in [1,r_u], h_\ell \in [1,g_\ell]$ and 
$i\in\mathcal{V}_1^{\ell,h_\ell}\cup \mathcal{V}_2^{\ell,h_\ell}$ if and only if 
the matrix $\Phi_u$ is Schur, which is equivalent to the asymptotic stability,  for every $\ell\in [1,r_u]$, $ h_\ell \in [1,g_\ell]$ and $i\in\mathcal{V}_1^{\ell,h_\ell}\cup \mathcal{V}_2^{\ell,h_\ell}$, of the  equations
\begin{align}
       e_{i}^{\ell,h_\ell}(t+1) =&(1 - k^{\ell,h_\ell}[\mathcal{L}]_{i,i})
         A^{\ell,h_\ell}e_{i}^{\ell,h_\ell}(t)\\
         &- k^{\ell,h_\ell}
        A^{\ell,h_\ell}
        \!\!\!\!\!\!\sum_{j\in (\mathcal{V}_{1}^{\ell,h_\ell} \cup \mathcal{V}_{2}^{\ell,h_\ell})} \!\!\![\mathcal{L}]_{i,j}e_j^{\ell,h_\ell}(t), \nonumber
\end{align} 
or, equivalently, of the linear system
\begin{align} \label{eq.error}
e_{u}^{\ell,h_\ell}(t+1)=\left(\left(I - k^{\ell,h_\ell}\mathcal{ L}^{\ell,h_\ell}\right)\otimes A^{\ell,h_\ell}\right) e_{u}^{\ell,h_\ell}(t),
\end{align}
where  
$e_{u}^{\ell, h_\ell} \doteq {\rm col} \{ e_i^{\ell, h_\ell}(t)\}_{i \in\mathcal{V}_1^{\ell,h_\ell}\cup \mathcal{V}_2^{\ell,h_\ell}}$ and 
 $\mathcal{L}^{\ell,h_\ell}$ is the matrix obtained from the Laplacian ${\mathcal L}$ by deleting the rows and columns   indexed in  $\mathcal{V}_3^{\ell,h_\ell}$.

The following theorem provides    necessary and sufficient conditions for the asymptotic stability of   system \eqref{eq.error} for all values of $\ell\in [1,r_u]$ and $ h_\ell \in [1,g_\ell]$, and hence for the existence of a distributed observer of the form \eqref{eq.lunberger.directed} and \eqref{eq.xiuhat}.
\begin{theorem}\label{thmain}
    Consider a system described by \eqref{eq.sys} and \eqref{eq.out}, with communication digraph $\mathcal{G}=(\mathcal{V},\mathcal{E}, \mathcal{A})$.  Let Assumptions \ref{ass.jointdet} and \ref{ass.obs} hold. There exists a distributed observer of the form \eqref{eq.lunberger.directed} and \eqref{eq.xiuhat} if and only if  for all $(\ell,h_\ell)\in [1,r_u]\times [1, g_\ell]$  there exists $k^{\ell,h_\ell}\in\mathbb{R}$ such that
    \begin{equation}|1 - k^{\ell,h_\ell}\mu|<\frac{1}{|\lambda_\ell|},
    \qquad \forall\ \mu \in\sigma\left(\mathcal{L}^{\ell,h_\ell}\right).
    \label{eq.mainineq}
    \end{equation}
\end{theorem}
\begin{proof}
We preliminarily observe that Assumption \ref{ass.obs} ensures that each pair $(A_{i,d}, C_{i,d}), i\in [1,N],$ is detectable.
Therefore, there exist a matrix $L_{i,d}$ that makes  $A_{i,d}-L_{i,d} C_{i,d}$  Schur, and hence $\Phi_d$ is a Schur matrix. This guarantees that each observer \eqref{eq.lunberger.directed} asymptotically estimates $z_{i,d}(t)$ and therefore $x_{i,d}(t)$ for every $i\in [1, N]$. Therefore a distributed observer of the form \eqref{eq.lunberger.directed} and \eqref{eq.xiuhat} exists if and only if we are able to ensure that for every $(\ell,h_\ell)\in [1,r_u]\times [1, g_\ell]$ the estimation error in \eqref{eq.error} asymptotically converges to zero. But this is equivalent to guaranteeing the existence, for every $\ell\in [1,r_u]$ and $h_\ell\in [1, g_\ell],$  of a real coefficient $k^{\ell, h_\ell}$ that makes the matrix $\left(I -  k^{\ell,h_\ell}\mathcal{L}^{\ell,h_\ell}\right)\otimes A^{\ell,h_\ell}$ Schur stable.
By Lemma \ref{prop.eigs.dir}, in the Appendix, this is possible if and only if  for every $(\ell,h_\ell)\in [1,r_u]\times [1, g_\ell]$  there exists $k^{\ell,h_\ell}\in\mathbb{R}$ 
such that \eqref{eq.mainineq} holds.
\end{proof}
\begin{remark}\label{comparison}
It is worthwhile at this point to compare in detail the distributed observer proposed in this paper with the original one in \cite{GaoYang} and, in particular, Theorem \ref{thmain} with Theorem 2 in \cite{GaoYang}.
The main advantage of the  distributed observer proposed in this article is that it is simpler than the original one in that it focuses on the estimate of the subvectors of the state vector corresponding to entire Jordan miniblocks, rather than considering the estimates of portions of such subvectors.
So, the main distinction is between agents that can estimate the whole subvector $x^{\ell, h_\ell}(t)$ and agents that cannot. This is the   criterion that leads to split, for each $i$th agent, the estimate $\hat x_i(t)$ in the two parts $\hat x_{i,d}(t)$ and  $\hat x_{i,u}(t)$.\\  
As a consequence of this choice, the solvability conditions of the problem are considerably simplified for two main reasons. First, the number of inequalities to be verified is greatly reduced: instead of one inequality per entry, we now have a single inequality \eqref{eq.mainineq} for each miniblock. This also means that the number of matrices—obtained by removing specific rows and columns from the original Laplacian ${\mathcal L}$—whose spectra must be computed is significantly lower. Second, rather than imposing a single coupling strength $k$
 for all observers, we can relax the existence conditions by allowing a distinct coupling strength for each pair $(\ell, h_\ell)$. For all these reasons, condition \eqref{eq.mainineq} is less restrictive than condition (30) in \cite{GaoYang}.\\
Finally, condition C2 in     Theorem 2 of \cite{GaoYang} proves to be implicit in \eqref{eq.mainineq}. Indeed, the existence of $k^{\ell, h_\ell}$ satisfying \eqref{eq.mainineq} rules out the possibility that $\mu=0$ is an eigenvalue of ${\mathcal L}^{\ell, h_\ell}$. So, condition \eqref{eq.mainineq} holds true only if (see Lemma 4 in \cite{GaoYang}) for every $(\ell, h_\ell)$ for which ${\mathcal V}_3^{\ell, h_\ell} \ne {\mathcal V}$ it holds what follows: every node indexed in ${\mathcal V}\setminus {\mathcal V}_3^{\ell, h_\ell}$ is reachable from at least one of the nodes indexed in ${\mathcal V}_3^{\ell, h_\ell}$ by means of a   directed path, namely there exists a ``directed spanning forest" rooted in the agents indexed in ${\mathcal V}_3^{\ell, h_\ell}$.
However, we remark again that compared to condition C2 in Theorem 2 of \cite{GaoYang} this communication constraint is imposed on a (typically much) smaller number of sets. 
\end{remark}
\section{Numerical Example}\label{sec.example}
The following example illustrates the proposed approach and also shows that the simplifying assumptions that the system matrix $A$ is already in Jordan form and has all real eigenvalues can be easily removed, to make the analysis completely general.\\ 
Consider a group of $N=6$ Pendubots in a formation, where the goal of each Pendubot is to estimate the joint angles of all Pendubots. Following \cite{GaoYang},  the dynamics of the $i$-th Pendubot  is modeled as  
\begin{equation*}
\zeta_i(t+1)=G\zeta_i(t)+Hu_i(t),\ i\in[1,6],
\end{equation*}
where $\zeta_i(t)\in {\mathbb R}^4$, and the entries  $[\zeta_{i}(t)]_1$ and $[\zeta_{i}(t)]_3$  represent the joint angles, while the system matrices are  
\begin{gather*}
\scriptsize G=\left[\begin{array}{cccc}0.9992& 0.005& 0.0003 &0\\
    -0.3369 &0.9992 &0.1242 &0.0003\\
    0.0008& 0 &1.0007 &0.005\\
    0.3263& 0.0008 &0.2786 &1.0007\end{array}\right],
     H=\left[\begin{array}{c} 0.0006\\0.2243\\-0.0001\\-0.0232\end{array}\right].
\end{gather*}
The overall dynamics of all Pendubots are described by a discrete-time linear time-invariant  system \eqref{eq.sys} with $x(t)\doteq {\rm col} \{\zeta_i(t)\}_{i\in[1,6]}$, $u(t)\doteq{\rm col}\{u_i(t)\}_{i\in[1,6]} $, $A\doteq I_6\otimes G$ and $B\doteq I_6\otimes H$.\\ 
The communication   among the Pendubots is modeled by the unweighted digraph $\mathcal{G}$ in Fig. \ref{graph}, whose Laplacian is
\[
\mathcal{L}=\scriptsize \left[\begin{array}{cccccc}
1 &0& 0& -1& 0 &0\\
    -1& 1& 0& 0& 0& 0\\
    -1& 0 &1 &0& 0& 0\\
     -1 &0& 0& 1 &0& 0\\
    0 &-1& 0 &-1 &2 &0\\
    0 &0& -1 &-1& 0& 2
\end{array}\right].
\]

\begin{figure}[tp!]
  \centering
\includegraphics[width=0.13\textwidth]{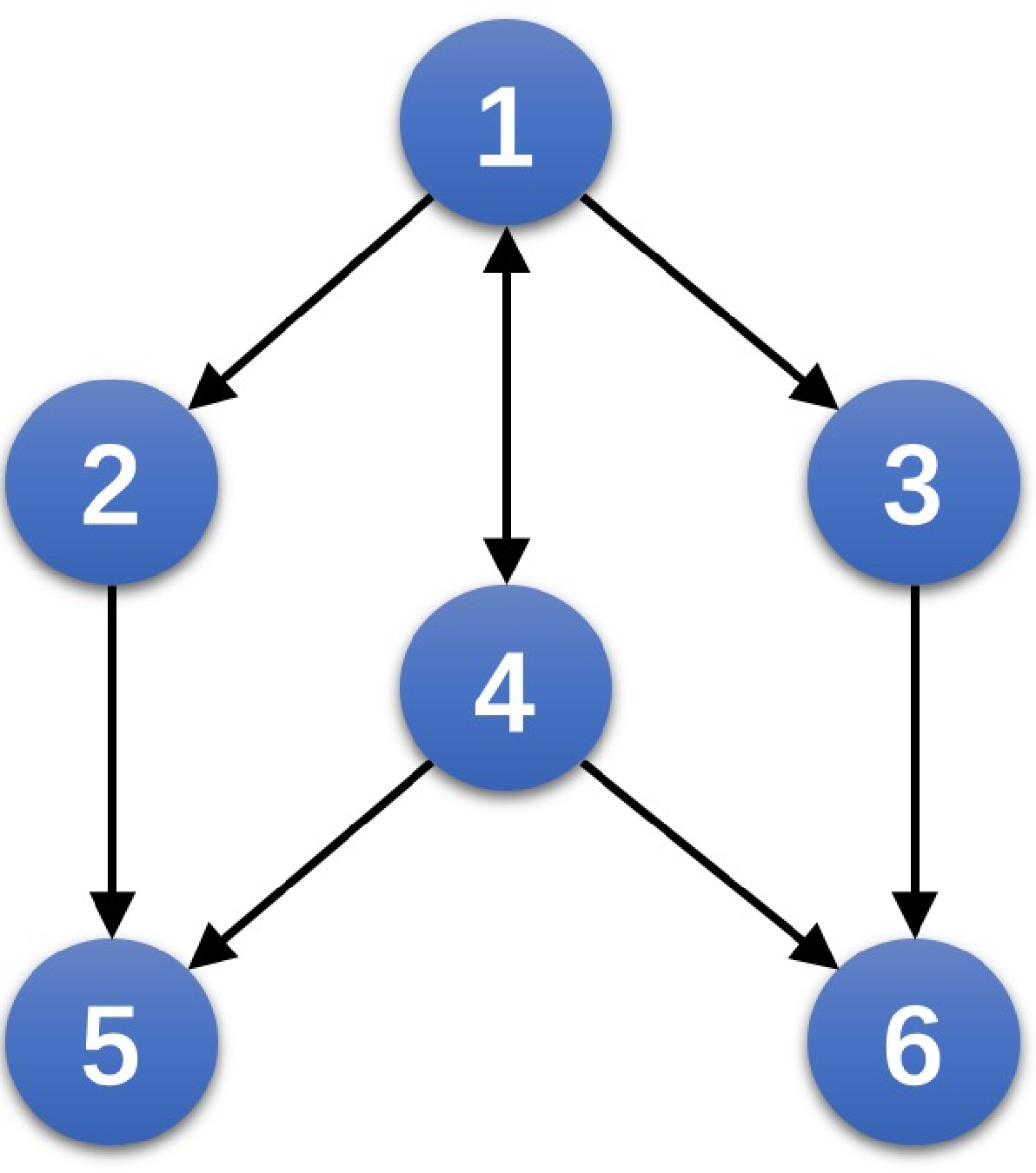}
  \caption{The communication digraph $\mathcal{G}$.}\label{graph}
\end{figure}
\noindent The measurement matrices of the Pendubots are
\[
\setlength{\arraycolsep}{2.6pt}%
\begin{split}
\scriptsize C_1=&\ \scriptsize\left[\begin{array}{cccccc}C &\mathbb{0} &\mathbb{0} &\mathbb{0}&\mathbb{0} &\mathbb{0}\\ C& -C &\mathbb{0} &\mathbb{0}&\mathbb{0} &\mathbb{0}\\C&\mathbb{0} & -C&\mathbb{0}&\mathbb{0} &\mathbb{0}\\ C&\mathbb{0} & \mathbb{0} & -C&\mathbb{0} &\mathbb{0}\end{array}\right],
C_4= \scriptsize \left[\begin{array}{cccccc}\mathbb{0} & \mathbb{0} &\mathbb{0} &C&\mathbb{0} &\mathbb{0}\\ -C& \mathbb{0}&\mathbb{0} &C&\mathbb{0} &\mathbb{0}\\\mathbb{0}&\mathbb{0} & \mathbb{0}&C&-C &\mathbb{0}\\ \mathbb{0}&\mathbb{0} & \mathbb{0} & C&\mathbb{0} &-C\end{array}\right],\\
C_2=&\ \scriptsize\left[\begin{array}{cccccc}\mathbb{0} & C &\mathbb{0} &\mathbb{0}&\mathbb{0} &\mathbb{0}\\ \mathbb{0}& C &\mathbb{0} &\mathbb{0}&-C &\mathbb{0}\end{array}\right],\quad   
C_3= \scriptsize\left[\begin{array}{cccccc}\mathbb{0} & \mathbb{0} &C &\mathbb{0}&\mathbb{0} &\mathbb{0}\\ \mathbb{0}& \mathbb{0} &C &\mathbb{0}&\mathbb{0}&-C\end{array}\right],
\\
C_5=&\ \scriptsize\left[\begin{array}{cccccc}\mathbb{0} & \mathbb{0} &\mathbb{0} &\mathbb{0}&C &\mathbb{0} \end{array}\right],\quad 
C_6= \scriptsize\left[\begin{array}{cccccc}\mathbb{0} & \mathbb{0} &\mathbb{0} &\mathbb{0}&\mathbb{0} &C\end{array}\right],
\end{split}
\]
where $C=\left[\begin{smallmatrix}1&0&0&0\\0&0&1&0\end{smallmatrix}\right]$. \\
First, we compute the eigenvalues of $A$, obtaining 
$$\lambda_1= 0.9991\pm 0.0445\mathrm{i},\quad \lambda_2=1.042,\quad  \lambda_3=0.9597$$ with $a_\ell=g_\ell=6$, for every ${\ell}\in[1,3]$, we verify  that Assumption \ref{ass.jointdet} holds and apply the nonsingular matrix $T=[I_6\otimes T_1\ I_6\otimes T_2\ I_6\otimes T_3]$ 
\begin{gather*}
 \setlength{\arraycolsep}{2pt} 
\scriptsize 
T_1=\left[\begin{array}{cc}0 &  -0.2324\\-2.0702 &0\\0  &  0.1122\\ 1&0\end{array}\right],\
 T_2=\left[\begin{array}{c}0.0223  \\ 0.1839\\0.1215  \\ 1\end{array}\right],\ T_3=\left[\begin{array}{c}-0.0224 \\ 0.1839\\-0.1215 \\ 1\end{array}\right].
\end{gather*}
to bring the matrix $A$ into a Jordan form of the type \eqref{Jform}-\eqref{Jminiblock}.
Subsequently, following the procedure in Sections \ref{sec.pb.setup} and \ref{sec.pb.solution}, we determine the permutation matrices $P_i$, $Q_i$  and $R_i$, $i\in[1,6]$, and we verify that Assumption \ref{ass.obs} holds. Due to space limitations, the matrices are not reported here. Then, for each node $i\in[1,6]$ we select $L_{i,d}$ such that $F_{i,d} - L_{i,d} H_{i,d}$ is Schur stable.

In order to perform the design of the coupling gains, we identify the nodes for which an entire miniblock is observable, namely $\mathcal{V}_{3}^{1,1}=\mathcal{V}_{3}^{2,1}=\{1,4\}$, $\mathcal{V}_{3}^{1,2}=\mathcal{V}_{3}^{2,2}=\{1,2\}$, $\mathcal{V}_{3}^{1,3}=\mathcal{V}_{3}^{2,3}=\{1,3\}$,  $\mathcal{V}_{3}^{1,4}=\mathcal{V}_{3}^{2,4}=\{1,4\}$,  $\mathcal{V}_{3}^{1,5}=\mathcal{V}_{3}^{2,5}=\{2,4,5\}$, and $\mathcal{V}_{3}^{1,6}=\mathcal{V}_{3}^{2,6}=\{3,4,6\}$. Based on this, we can compute the eigenvalues of  $\mathcal{L}^{\ell,h_\ell}$  for $\ell\in[1,2]$ and ${h_\ell}\in[1,6]$. Then, by solving the inequalities \eqref{eq.mainineq}, we obtain the ranges of $k^{\ell,h_\ell}$, $\ell\in[1,2]$, and $h_\ell\in[1,6]$,
\[
\begin{split}
0<k^{1,h_\ell}<1,\quad 
0.0403<k^{2,h_\ell}<0.9798.
\end{split}
\]
Once we assume $k^{1,1}=k^{2,1}=0.2$, $k^{1,2}=k^{2,2}=0.3$, $k^{1,3}=k^{2,3}=0.4$, $k^{1,5}=k^{2,5}=0.5$, $k^{1,6}=k^{2,6}=0.6$, $k^{1,7}=k^{2,7}=0.7$,  the gain parameters $K_i$ in \eqref{eq.xiuhat} are determined. \smallskip

The simulation results are plotted in Fig. \ref{x}. We can see that  each Pendubot can asymptotically estimate the joint angles of all Pendubots. This shows the effectiveness of the  proposed  method. 
\begin{figure}[tp!]
  \centering
  \hspace{-1pt}\includegraphics[width=0.24\textwidth]{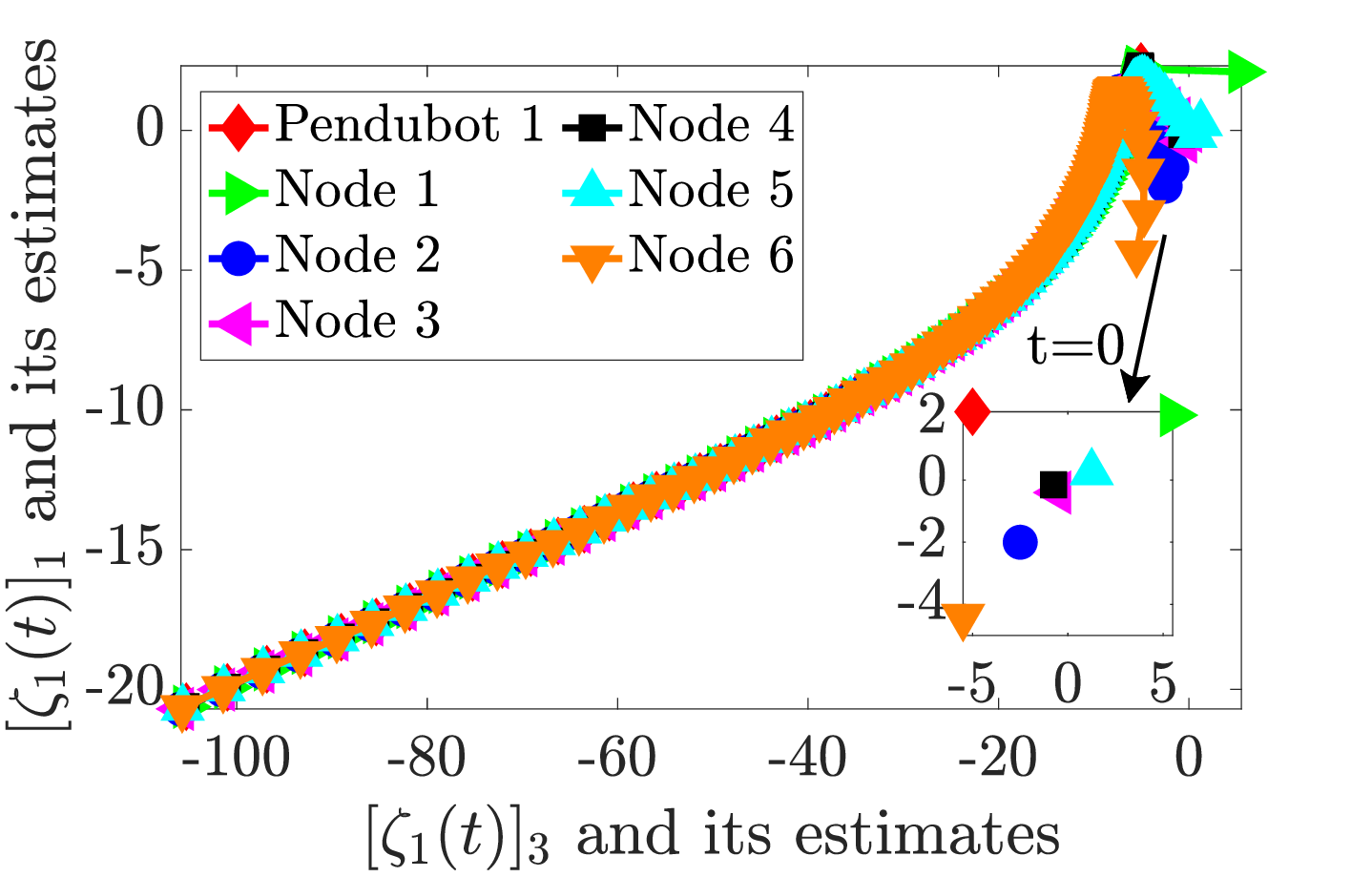}
  \includegraphics[width=0.24\textwidth]{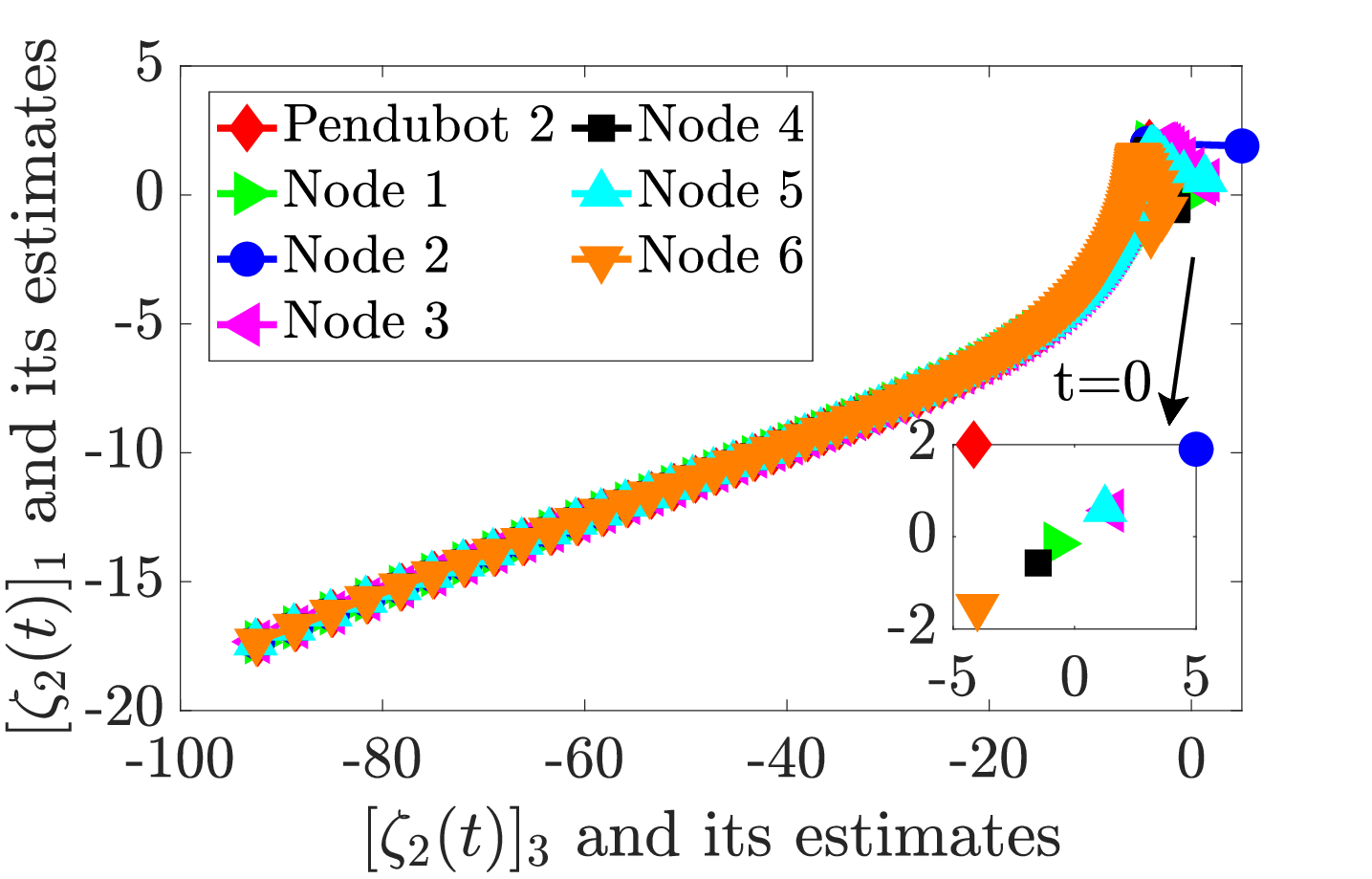}\\
 \hspace{-4pt} \includegraphics[width=0.24\textwidth]{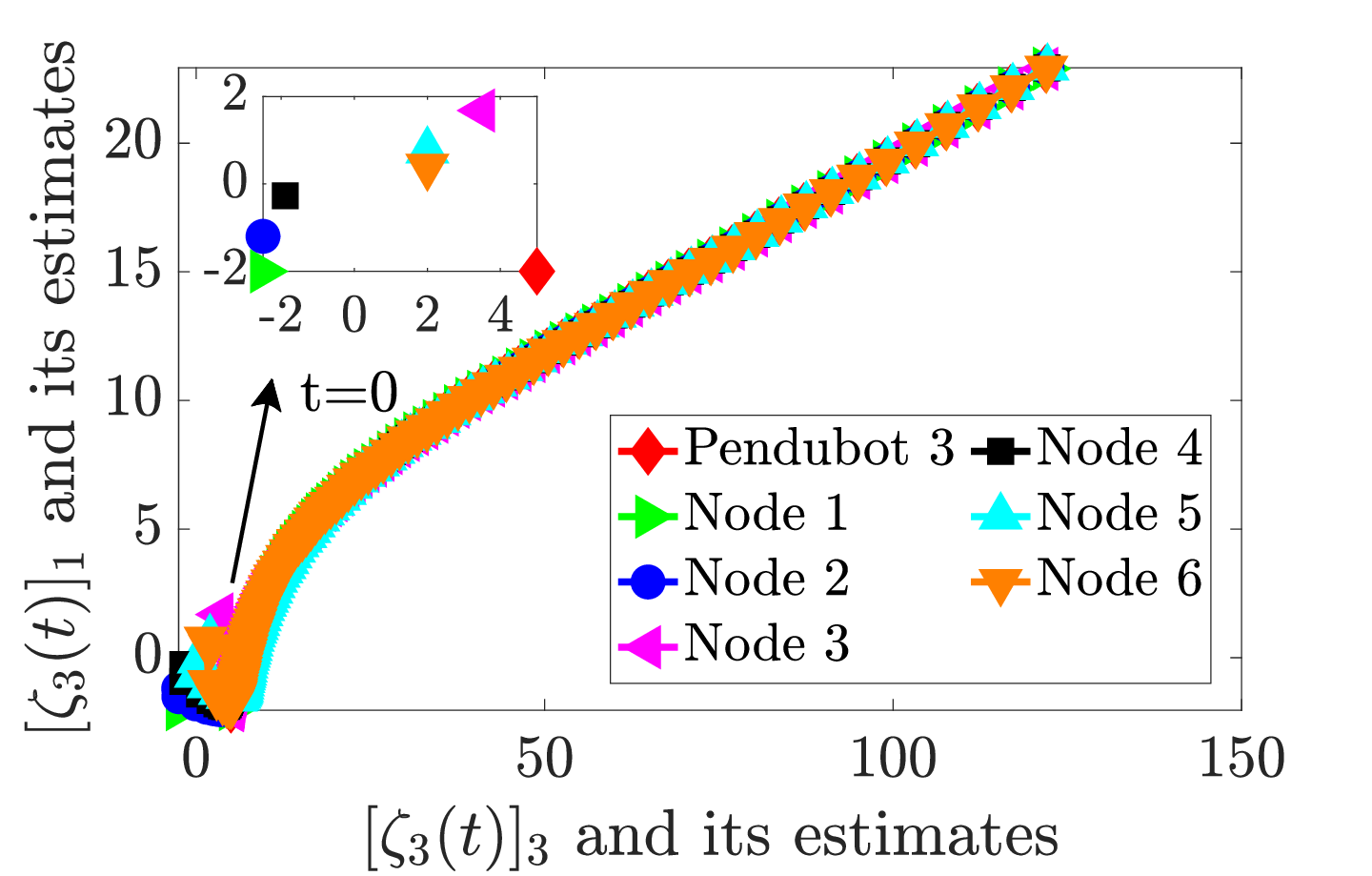}
  \includegraphics[width=0.24\textwidth]{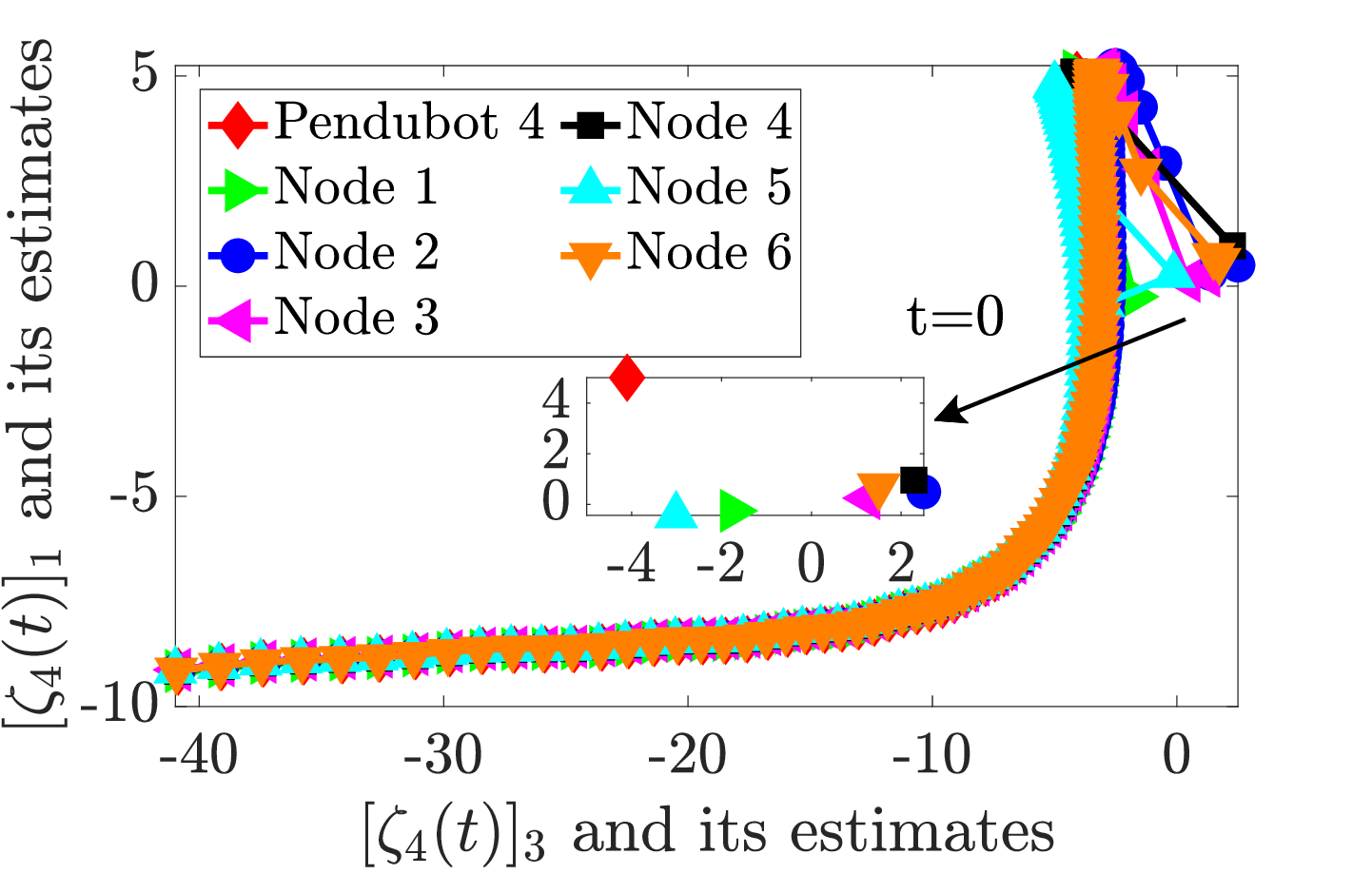}\\
 \hspace{-1pt}\includegraphics[width=0.24\textwidth]{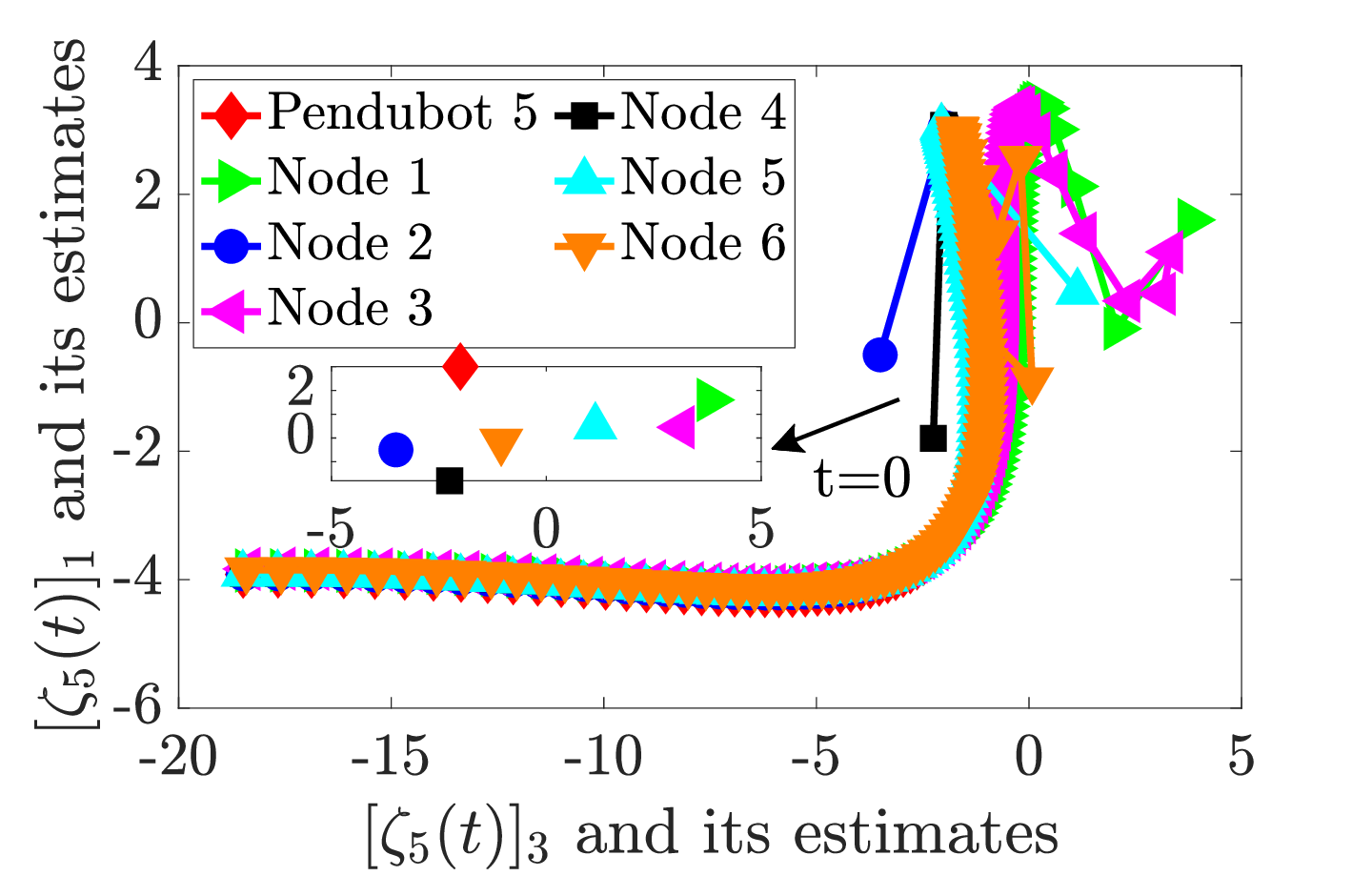}
\includegraphics[width=0.24\textwidth]{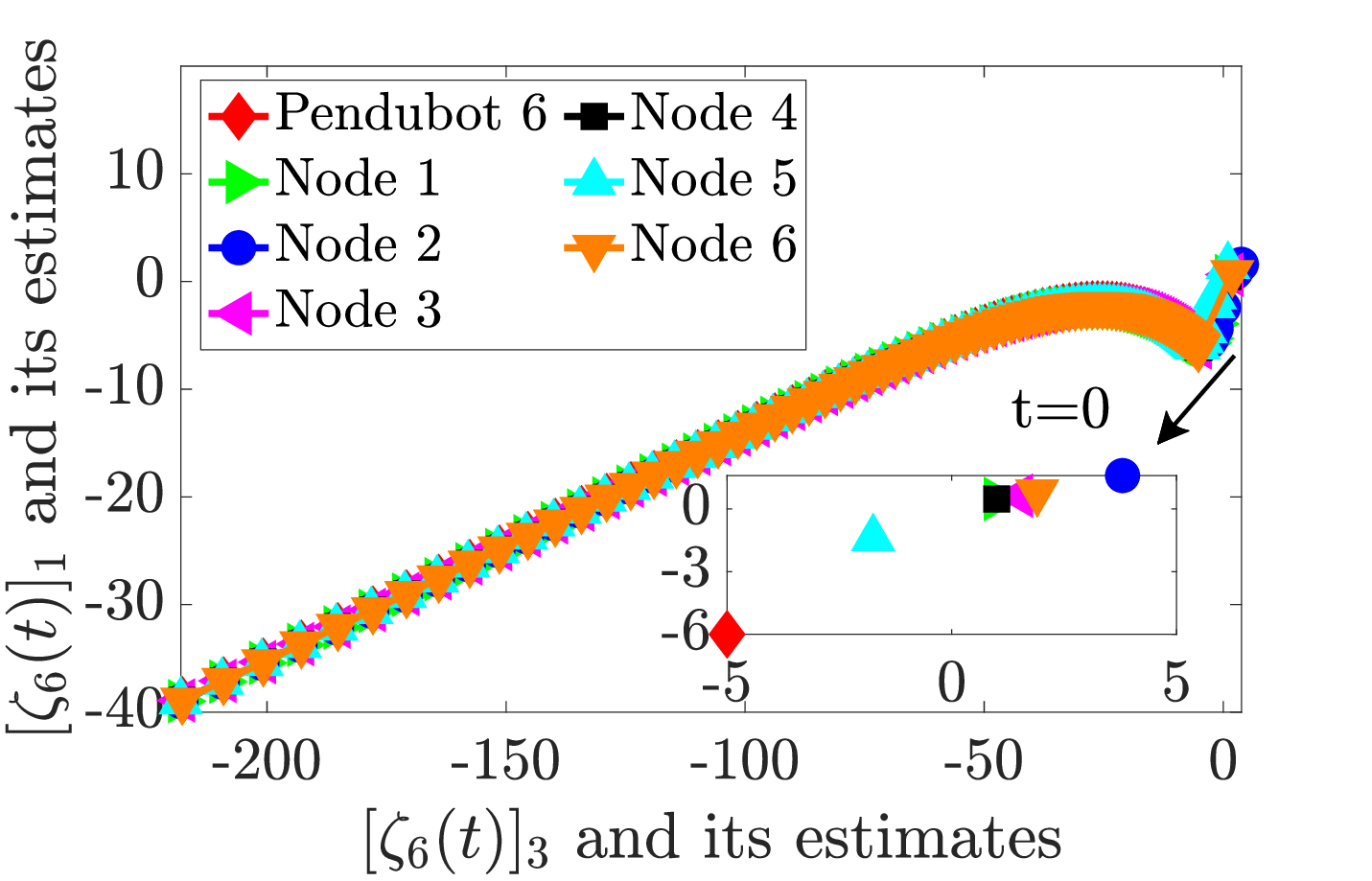}
  \caption{The joint angles of the Pendubot and their estimates.}\label{x}
\end{figure}
\section{Conclusions}\label{sec.conclusion}
In this paper, we have proposed a distributed state estimation scheme that leverages both the Jordan canonical form of the system matrix $A$ and the standard (Kalman) detectability form of each pair $(A,C_i)$. We first introduced the Jordan canonical decomposition and showed that, under suitable assumptions, there exists a permutation matrix that brings all the pairs of system and output matrices into standard detectability form.  Each $i$th agent asymptotically estimates the part of the state  that is detectable for it, $z_{i,d}(t)$, using a Luenberger observer, but of the vector $\hat z_{i,d}(t)$ only the blocks $\hat x_i^{\ell, h_\ell}(t)$ corresponding to Jordan miniblocks that are completely detectable are retained in $\hat x_{i,d}(t)$. Then agent $i$ estimates the remaining part of the state, $x_{i,u}(t)$, through a consensus-based algorithm, with a specific coupling gain for each Jordan miniblock. Necessary and sufficient conditions for the solvability of the distributed estimation problem were derived. A detailed comparison between the conditions for problem solvability derived in this paper and those first proposed in \cite{GaoYang} was provide in Remark \ref{comparison}. Finally, a numerical example showed the effectiveness of the results.

Future work will focus on extending this approach to   systems affected by unknown inputs and disturbances, thereby further enhancing the robustness and applicability of the proposed method.

\appendix
\renewcommand{\thetheorem}{A.\arabic{theorem}} 
\setcounter{theorem}{0}

\begin{lemma}
\label{prop.eigs.dir}
Given a Jordan miniblock $A$ corresponding to some eigenvalue $\lambda\in {\mathbb R}, |\lambda|\ge 1$, a   matrix $L\in {\mathbb R}^{n\times n}$ and a scalar $k\in {\mathbb R}$, the following facts are equivalent.
\begin{itemize}
    \item[i)] The eigenvalues of $\left(I_n - k L\right)\otimes A$ have moduli smaller than $1$.
    \item[ii)] 
    For every $\mu \in\sigma(L)$, one has
    \begin{equation}
        \label{eq.diseq}
    |1 - k \mu|<\frac{1}{|\lambda|}.
    \end{equation}
\end{itemize}
\end{lemma}
\begin{proof}  
Let
$\sigma(I_n - k L) =
\{\alpha_1, \dots, \alpha_n\}$.
By the properties of the Kronecker product, 
$$\sigma\left((I_n - k L)\otimes A\right) =
\{\alpha_1 \lambda, \dots, \alpha_n \lambda\}.$$
    Let $T$ be a transformation matrix that reduces $L$ to Jordan form
    $$T^{-1} L T = J_L.$$
    Then it holds that
    \begin{align*}
        &I_n - kL =I_n - k T J_L T^{-1} 
        =T (I_n - k J_L)T^{-1}
    \end{align*}
which ensures that 
$\{\alpha_1, \dots, \alpha_n \} = \sigma(I_n - k L) = \sigma(I_n - k J_L) = \{ 1 - k \mu,  \mu\in \sigma(L)\}.$
\\
    This immediately proves that i) 
    holds, namely $|\alpha_i \lambda|<1$ for every $i\in [1,n]$, if and only if 
     \eqref{eq.diseq} holds for every $\mu\in \sigma(L)$, which is condition ii).
\end{proof}

\end{document}